\documentclass[]{IEEEtran}%
\usepackage{cite}
\usepackage{amsfonts}%
\usepackage[cmex10]{amsmath}
\usepackage[utf8]{inputenc}
\usepackage{fixmath}%
\usepackage{amssymb}%
\usepackage{graphicx}
\usepackage{comment}
\usepackage{url}
\usepackage{setspace}
\usepackage{mathtools}
\usepackage{subfig}
\usepackage{color}
\usepackage{tabularx}
\usepackage{mathtools}
\usepackage{amssymb}
\usepackage{amsthm}
\usepackage{mathtools}
\usepackage{graphicx}
\allowdisplaybreaks

\usepackage[cmex10]{amsmath}
\usepackage[utf8]{inputenc}
\usepackage[english]{babel}
\usepackage{amssymb}
\usepackage{amsthm}
\usepackage{mathtools}
\newtheorem{theorem}{Theorem}

\usepackage{algorithmic}

\usepackage{array}

\allowdisplaybreaks
\begin{document}
\title{Uplink Performance Analysis of Cell-Free mMIMO Systems under Channel Aging }
\author{\IEEEauthorblockN{Ribhu Chopra,~\IEEEmembership{Member, IEEE}, Chandra R. Murthy,~\IEEEmembership{Senior Member, IEEE}, and Anastasios K. Papazafeiropoulos,~\IEEEmembership{Senior Member, IEEE}} 
	
	\thanks{R. Chopra  is with the Department of Electronics and Electrical Engineering, Indian Institute of Technology Guwahati, Assam, India. (email: ribhufec@iitg.ac.in).			
		C. R. Murthy is with the Department of Electrical Communication Engineering, Indian Institute of
		Science, Bangalore, India.
		(email: cmurthy@iisc.ac.in).
	A. K. Papazafeiropoulos is with the CIS Research Group, University of Hertfordshire, Hatfield, U. K., and SnT at the University of Luxembourg, Luxembourg (e-mail: tapapazaf@gmail.com).

	}	
}

\maketitle
\begin{abstract}
	In this paper, we investigate the impact of channel aging on the uplink performance of a cell-free~(CF) massive multiple-input multiple-output (mMIMO) system with a  minimum mean squared error~(MMSE) receiver.  To this end, we present a new model for the temporal evolution of the channel, which allows the channel to age at different rates at different access points (APs). Under this setting, we derive the deterministic equivalent of the per-user achievable signal-to-interference-plus-noise ratio~(SINR). In addition to validating the theoretical expressions, our simulation results reveal that, {at low user mobilities,} the  SINR of CF-mMIMO is nearly  $5$ dB higher than its cellular counterpart with the same number of antennas, and about $8$ dB higher than that of an equivalent small-cell network with the same number of APs. {On the other hand, at very high user velocities, and when the channel between the UEs the different APs age at same rate, the relative impact of aging is higher for CF-mMIMO compared to cellular mMIMO. However, when the channel ages at the different APs with different rates,  the effect of aging on CF-mMIMO is marginally mitigated, especially for larger frame durations.} 
\end{abstract}
\section{Introduction}
Cell-free~{(CF)} massive multiple-input multiple-output~{(mMIMO)} {systems are considered}  the natural successor to the cellular mMIMO {technology} for physical layer access in next-generation wireless  systems~\cite{BJORNSON20193,Ngo_TWC_2019,Red_book,Masoumi_TWC_2020}. The canonical CF-mMIMO setup consists of a large number of access points~(APs) spread over a given physical area, and connected to a single central processing unit~(CPU). Since the signals received from all the UEs at multiple APs are processed jointly at the CPU, a CF system becomes a distributed mMIMO system. In contrast, a cellular mMIMO system consists of a single base station/AP with a massive number of antennas serving all the users in the cell. 
CF-mMIMO systems offer the advantage of more uniform coverage  compared to conventional cellular systems, while retaining the benefits of mMIMO systems such as high spectral efficiency~\cite{bounds} and simple linear processing at the APs/CPU~\cite{Red_book}. However, CF-mMIMO also inherits mMIMO's  dependence on accurate channel state information~(CSI)~\cite{Papa_TWC_2015,CRM_TWC_2017}. However, the available CSI may be impaired due to several factors such as pilot contamination~\cite{EGL_TWC_2019}, channel aging~\cite{aging1}, and {hardware impairments~\cite{Masoumi_TWC_2020}.}  

The phenomenon of channel aging, classically known as time-selectivity, is a manifestation of the temporal variation in the wireless channel  caused due to user mobility and phase noise~\cite{Mahyiddin_ComLet_2015, aging1,Papa_TWC_2015,Papa_TVT_2016,bounds,Lett_aging_LMIMO,CRM_TWC_2017}. It has been shown that while the power scaling and the  array gains achieved by mMIMO are retained under channel aging~\cite{bounds}, {the achievable signal-to-interference-plus-noise ratio (SINR)} {gradually decays} over time {as the channel estimates become outdated}, which, in turn, limits the system dimensions~\cite{Lett_aging_LMIMO,CRM_TWC_2017}. {While there have been recent efforts to counter the effect of aging via channel prediction~\cite{guo2020predictor}, the applicability of these techniques to cellular or cell free mMIMO systems remains to be seen.}

In terms of the effect of aging on the underlying architecture, a key difference between cellular and CF  mMIMO systems is that  
in the latter system, the channels between the APs and UEs could potentially evolve over time at different rates. 
 {Together with} the CPU based joint processing of the uplink signals, this makes the analysis (and ultimately the effect) of channel aging on CF systems {different from cellular mMIMO~\cite{Papa_TVT_2016}.} {However, despite its significance}, to the best of our knowledge, the effect of channel aging on CF-mMIMO systems has not been investigated in the previous literature.

In this paper, we characterize the effect of channel aging on the uplink achievable SINR of a CF-mMIMO system. First, we develop a model for the relative speed of the users with respect to the different APs. Then, using results from random matrix theory~\cite{RMT}, we derive {an analytical expression for the deterministic equivalent of the SINR, i.e., the SINR averaged over all fading channel realizations, as a function of the UE locations and velocities}. Through simulations, we compare the achievable SINR against (a) the conventional cellular mMIMO {setup,} and (b) a small cell system~\cite{Ngo_TWC_2019}, to assess the relative impact of channel aging on the three systems. We conclude that CF-mMIMO systems generally provide significantly better SINR compared to the other two systems in the presence of channel aging, but {the relative impact of aging is more severe on CF-mMIMO at very high user velocities and long frame durations.} 
\vspace{-10pt}
\section{System Model and Channel Estimation}
We consider a CF-mMIMO system with $M$ APs each having $N$ antennas and serving a total of $K$ users. The APs are connected to a CPU over a delay-free unlimited capacity front-haul link. The locations of the APs and the UEs are considered to be points on the two dimensional plane, with the locations of the $m$th {AP and the $k$th UE represented by using the complex numbers $\mathtt{d}_{a,m}$ and $\mathtt{d}_{u,k}$, respectively}. 
{Also, } the path loss between the $k$th user and the $m$th AP is modeled {by} using a piecewise linear multi-slope model as
\begin{equation}
\beta_{mk}= \mu_l\left({\lvert \mathtt{d}_{a,m}-\mathtt{d}_{u,k} \rvert}\right)^{-\eta_l},  d_{l-1}<\lvert \mathtt{d}_{a,m}-\mathtt{d}_{u,k} \rvert < d_l, \label{eq:multi_slope}
\end{equation}
{where $\mu_l$ is a normalization constant,  $\eta_{l}$  is the $l$th slope, and  $d_l$ is the $l$th threshold,} $1 \le l \le L$, with $d_L = \infty$~\cite{multi_slope}.

The fast fading channel between the $k$th user and the $m$th AP at the $n$th {instant}, denoted by $\mathbf{h}_{mk}[n]\in\mathbb{C}^{N\times 1}$, is assumed to evolve in time according to the relation~\cite{Lett_aging_LMIMO, Jakes}
\begin{equation}
\mathbf{h}_{mk}[n+\tau]=\rho_{mk}[\tau]\mathbf{h}_{mk}[n]+\bar{\rho}_{mk}[\tau]{\mathbf{z}_{h,mk}[n+\tau;n]}, \label{eq:channel_evolution}
\end{equation}
where $\rho_{mk}[\tau]$ is the  correlation coefficient of the channel between the $k$th UE and the $m$th AP at lag $\tau$. 
 The channel state $\mathbf{h}_{mk}[n]$ and the innovation $\mathbf{z}_{h,mk}[n+\tau;n]$ are assumed to be distributed as $\mathcal{CN}(0,\mathbf{I}_N)$, with $E[\mathbf{h}_{mk}[n]\mathbf{z}^H_{h,mk}[n+\tau;n]]=\mathbf{O}_{N}$, {where} $\mathbf{I}_N$ and $\mathbf{O}_N$ are the $N\times N$ identity and all zero matrices, respectively. We also assume that the innovation process $\mathbf{z}_{h,mk}[t;n]$ is wide sense stationary, and temporally white over the time index $t$, for a given anchor index~$n$. In \eqref{eq:channel_evolution} and throughout the paper, for any variable $x \in [-1, 1]$, {we denote} $\bar{x} \triangleq \sqrt{1 -x^2}$. 
 
 Conventionally, the correlation coefficient $\rho_{mk}[\tau]$ is assumed to follow the Jakes' model~\cite{Jakes}, or the first order autoregressive~(AR-1) model~\cite{Lett_aging_LMIMO}.  It is dependent on the user's speed through the Jakes' spectrum.  {However, this model is developed under the assumption that the scatterers are localized around the UEs. In case of CF-mMIMO systems, in general, it is not reasonable to expect that scatterers are localized only around the UE, and assume that the correlation coefficient is the same across all APs. {Moreover}, to the best of our knowledge, measurement campaigns elucidating the variation of the temporal correlation coefficient across APs are not available in the literature. Therefore, in this paper, we consider a simple generalization of the existing Jakes' spectrum based model, where the correlation coefficient depends on the relative speed of the UEs with respect to the different APs. To elaborate, we consider $v_{mk}$ to be a random variable that is i.i.d. across the APs with mean $v_{k}$ (i.e., depending only on the user index) and support $[(1-\Delta) v_{k},  (1+\Delta) v_{k}]$, where $0 \le \Delta \le 1$. The correlation coefficient is then defined as $\rho_{mk}[\tau] = J_0(2 \pi f_{d,mk} T_s\tau)$, where $J_0(\cdot)$ is the Bessel function of zeroth-order and first kind\cite{Jakes}, $f_{d,mk}= v_{mk} f_c/c$ is the Doppler frequency of the $k$th user with respect to the $m$th AP, $f_c$ is the carrier frequency, $c$ is the speed of light, and $T_s$ is the sampling interval. }
 
{Note that  $\Delta = 0$ results in the conventional model, where the correlation coefficient is the same at all APs.} 
{Nonetheless, even in this case, the effect of aging on a CF-mMIMO is different from cellular mMIMO. This is because, in cellular mMIMO, only the associated base station decodes the signal from a given UE, and the UE's signal arriving at other base stations is treated as interference. When $\Delta > 0$, the channel ages at different APs at different rates, and we will see in the section on numerical results that this has a marginally beneficial effect on a CF system at high user velocities and large frame durations.}



The uplink frame, consisting of a total of $T$ channel uses, is segmented into two subframes: training and data transmission. Over the first $P(\le K)$ channel uses, the UEs transmit uplink pilots to the APs with the $l$th UE employing a pilot energy $\mathcal{E}_{upl}$. For simplicity, we consider that the individual pilot signals occupy a single time slot, and the UEs whose pilots share the $p$th slot are contained in the set $\mathcal{U}_{p}$. These pilots are used by the respective APs to estimate the channels from these UEs.
In the next phase, consisting of $T-P$ channel uses, the UEs transmit uplink data. 
In this paper, we assume the system to have a ``level 4'' centralization of data processing according to the classification in~\cite{Bjornson_TWC_2020}. That is, the APs share all the available information with the CPU, including the channel estimates and the received data symbols. At the CPU, the received symbols are processed using an MMSE combiner to obtain estimates of the symbols transmitted by the users.

During the $p$th uplink training slot, the signal received by the $m$th AP is given by
\begin{align}
&\mathbf{y}_{m}[p]=\sum_{l \in \mathcal{U}_p}\sqrt{\beta_{ml} \mathcal{E}_{upl}} \mathbf{h}_{ml}[p]+\sqrt{N_0}\mathbf{w}_{l}[p]\nonumber \\&=\sqrt{\beta_{mk} \mathcal{E}_{upk}} \mathbf{h}_{mk}[p]+\sum_{\substack{l \in \mathcal{U}_p\\l\neq k}}\sqrt{\beta_{ml} \mathcal{E}_{upl}} \mathbf{h}_{ml}[p]+\sqrt{N_0}\mathbf{w}_{l}[p].\nonumber
\end{align}

The $m$th AP uses the pilot signal received at time $p$, $\mathbf{y}_{m}[p]$, to obtain the MMSE estimate of the $k$th UEs channel at time $P$, i.e., $\mathbf{h}_{mk}[P]$. We denote the MMSE estimate of $\mathbf{h}_{mk}[P]$ by $\mathbf{\hat{h}}_{mk}$. The estimate $\hat{\mathbf{h}}_{mk}$ is used by the CPU to decode the $k$th UE's signal over the entire frame. 

Using the well-known time-reversal property~\cite{OSAWA198861} of \eqref{eq:channel_evolution}, we can write  
\begin{align}
\mathbf{y}_{m}[p]&=\rho_{mk}[P-p]\sqrt{\beta_{mk} \mathcal{E}_{upk}} \mathbf{{h}}_{mk}[P]\nonumber\\&+\bar{\rho}_{mk}[P-p]\sqrt{\beta_{mk} \mathcal{E}_{upk}} \mathbf{{z}}_{h,mk}[p;P]\nonumber\\&+\sum_{\substack{l \in \mathcal{U}_p\\l\neq k}}\sqrt{\beta_{ml} \mathcal{E}_{upl}} \mathbf{h}_{ml}[p]+\sqrt{N_0}\mathbf{w}_{l}[p].
\end{align}
Consequently, we have
$
\mathbf{\hat{h}}_{mk}=\frac{\rho_{mk}[{P-p}]\sqrt{\beta_{mk}\mathcal{E}_{upk}}}{\sum_{l\in \mathcal{U}_p}\beta_{ml}\mathcal{E}_{upl}+N_0}\mathbf{y}_{m}[p], 
$
and
\begin{equation}
\mathbf{h}_{mk}[P]=a_{mk}\hat{\mathbf{h}}_{mk}+\bar{a}_{mk}\tilde{\mathbf{h}}_{mk},
\end{equation}
with $E[\hat{\mathbf{h}}_{mk}\tilde{\mathbf{h}}^H_{mk}]=\mathbf{O}_N$ {and} $a_{mk}=\sqrt{\frac{\rho^2_{mk}[P-p]\beta_{mk}\mathcal{E}_{upk}}{\sum_{l\in \mathcal{U}_p}\beta_{ml}\mathcal{E}_{upl}+N_0}}$. 

Therefore, for $P+1\le n \le T$, we have
\begin{align} 
\mathbf{h}_{mk}[n]=&\rho_{mk}[n-P]
a_{mk}\hat{\mathbf{h}}_{mk}+\rho_{mk}[n-P]\bar{a}_{mk}\tilde{\mathbf{h}}_{mk}\nonumber \\
+&\bar{\rho}_{mk}[n-P]\mathbf{z}_{h,mk}[n;P].
\label{eq:hmk}
\end{align}


With the system model and channel estimates in hand, we can now proceed with the uplink SINR analysis. 
\vspace{-10pt}
\section{Uplink SINR Analysis}

During the data phase, i.e., over the next $T-P$ channel uses, the UEs simultaneously transmit the data to all the APs. If the $k$th UE transmits the symbol $s_k[n]$ at the $n$th instant (satisfying $E[\lvert s_k[n] \rvert^2]=1$) with power $\mathcal{E}_{usk}$, then the signal received at the $m$th AP is 
\begin{equation}
\mathbf{y}_m[n]=\sum_{k=1}^K\sqrt{\beta_{mk}\mathcal{E}_{usk}}\mathbf{h}_{mk}[n]s_k[n]+\sqrt{N_0}\mathbf{w}_m[n].
\end{equation}
Letting $\mathbf{y}[n]\triangleq [\mathbf{y}^T_1[n], \mathbf{y}^T_2[n], \ldots, \mathbf{y}^T_M[n]]^T$, $\mathbf{h}_{k}[n]\triangleq [\mathbf{h}^T_{1k}[n], \mathbf{h}^T_{2k}[n], \ldots, \mathbf{h}^T_{Mk}[n]]^T$, $\mathbf{H}[n]\triangleq [\mathbf{h}_{1}[n], \mathbf{h}_{2}[n], \ldots, \mathbf{h}_{K}[n]]$, $\boldsymbol{\beta}_{k}\triangleq [\beta_{1k}, \beta_{2k}, \ldots, \beta_{Mk}]^T$, ${\mathtt{{B}}}\triangleq [\boldsymbol{\beta}_{1}, \boldsymbol{\beta}_{2}, \ldots, \boldsymbol{\beta}_{K}]$, $\boldsymbol{\rho}_{k}[n]\triangleq [\rho_{1k}[n], \rho_{2k}[n], \ldots, \rho_{Mk}[n]]^T$, ${\mathtt{{R}}}[n]\triangleq [\boldsymbol{\rho}_{1}[n], \boldsymbol{\rho}_{2}[n], \ldots, \boldsymbol{\rho}_{K}[n]]$, $\mathbf{a}_{k}\triangleq [a_{1k}, a_{2k}, \ldots, a_{Mk}]^T$, ${\mathbf{{A}}}\triangleq [\mathbf{a}_{1}, \mathbf{a}_{2}, \ldots, \mathbf{a}_{K}]$, $\mathbf{s}[n] \triangleq [s_1[n], \ldots, s_K[n]]^T$, we can write the concatenated signal received by the CPU at the $n$th instant as
\begin{align} 
&\mathbf{y}[n]=\sum_{k=1}^K\sqrt{\mathcal{E}_{usk}(\mathbf{1}_N\otimes\boldsymbol{\beta}_{k})}\odot\mathbf{h}_{k}[n]s_k[n]+\sqrt{N_0}\mathbf{w}[n] \nonumber \\
&=(\mathbf{1}_N\otimes(\mathbf{A}\odot\mathtt{R}[n-P]\odot\sqrt{\mathtt{B}}))\odot\mathbf{\hat{H}}\text{diag}(\sqrt{\boldsymbol{\mathcal{E}}_{us}})\mathbf{s}[n] \nonumber \\
&+(\mathbf{1}_N\otimes(\mathbf{\bar{A}}\odot\mathtt{R}[n-P]\odot\sqrt{\mathtt{B}}))\odot\mathbf{\tilde{H}}\text{diag}(\sqrt{\boldsymbol{\mathcal{E}}_{us}})\mathbf{s}[n] \nonumber \\
&+ (\mathbf{1}_N\otimes(\mathtt{\bar{R}}[n-P]\odot\sqrt{\mathtt{B}}))\odot\mathbf{Z}_h[n;P]\text{diag}(\sqrt{\boldsymbol{\mathcal{E}}_{us}})\mathbf{s}[n] \nonumber \\
&+ \sqrt{N_0}\mathbf{w}[n],
\label{eq:yn_expand}
\end{align}
{where} $\odot$ and $\otimes$ represent the Hadamard product and the Kronecker product of two matrices, respectively, and $\mathbf{1}_N$ {is} the $N\times 1$ all ones vector. The square-root operation above (and elsewhere in this paper) represents the element{-}wise square root; note that all the entries of these matrices are nonnegative. To obtain the last equation above, we use~\eqref{eq:hmk}, and anchor time at the $P$th instant. The first term in \eqref{eq:yn_expand} corresponds to the signal $s[n]$ received over the known channel $\mathbf{\hat{H}}$, the second term corresponds to the interference caused due to the channel estimation error at time $P$, the third term corresponds to the interference due to the evolution of the channel between time $P$ and time $n$, and the last term corresponds to AWGN. 

Now, the received vector is first pre-processed using an MMSE combiner. We will use the expression for the combined signal to find the achievable rate of the system via a deterministic equivalent analysis. The MMSE combiner requires the covariance matrix of the received signal{, conditioned on the available estimate $\hat{\mathbf{H}}$, which} can be expressed as
\begin{align}
&\mathbf{R}_{yy\vert \mathbf{\hat{H}}}[n]=\left(\mathbf{1}_N\otimes(\mathtt{R}[n-P]\odot \mathbf{A} \odot \sqrt{\mathtt{B}})\odot\hat{\mathbf{H}}\right)\text{diag}(\boldsymbol{\mathcal{E}})\nonumber\\
&\times\left(\mathbf{1}_N\otimes(\mathtt{R}[n-P]\odot \mathbf{A} \odot \sqrt{\mathtt{B}})\odot\hat{\mathbf{H}}\right)^H+(\mathbf{1}_N\mathbf{1}_N^T)\nonumber\\
&\otimes\Bigg(\!\left(\mathtt{R}[n-P]\odot \bar{\mathbf{A}}\!\odot \!\sqrt{\mathtt{B}}\!\right)\!\text{diag}(\boldsymbol{\mathcal{E}})\!\left(\mathtt{R}[n-P]\!\odot\! \bar{\mathbf{A}}\! \odot\! \sqrt{\mathtt{B}}\right)^T\nonumber\\
&+\left(\!\mathtt{\bar{R}}[n-P]\! \odot\! \sqrt{\mathtt{B}}\!\right)\!\text{diag}(\boldsymbol{\mathcal{E}})\left(\mathtt{\bar{R}}[n-P]\!\odot \! \sqrt{\mathtt{B}}\right)^T\!\Bigg)\!+\!N_0\mathbf{I}_{MN}.\nonumber
\\&=\mathbf{\hat{G}}\mathbf{\hat{G}}^H+\boldsymbol{\Psi}\otimes\mathbf{I}_N.
\end{align}
Here, for convenience, we have defined, $\hat{\mathbf{g}}_k[n]=(\mathbf{1}_N\otimes(\boldsymbol{\rho}_k[n-P]\odot\mathbf{a}_k\odot\sqrt{\mathcal{E}_{usk}\boldsymbol{\beta}_k})\odot\hat{\mathbf{h}}_k)$, $\mathbf{\hat{G}}=[\hat{\mathbf{g}}_1[n], \ldots, \hat{\mathbf{g}}_K[n]]^T$, and $\boldsymbol{\Psi}=\text{diag}(\psi_m)$, as a diagonal matrix with $\psi_m$ corresponding to the noise and interference power at the $m$th AP. Also defining $\tilde{\mathbf{g}}_k[n]=(\mathbf{1}_N\otimes(\boldsymbol{\rho}_k[n-P]\odot\bar{\mathbf{a}}_k\odot\sqrt{\mathcal{E}_{usk}\boldsymbol{\beta}_k})\odot\tilde{\mathbf{h}}_k)$, and $\boldsymbol{\xi}_{h,l}[n;P]=(\mathbf{1}_N\otimes(\boldsymbol{\bar{\rho}}_l[n-P]\odot\sqrt{\mathcal{E}_{usl}\boldsymbol{\beta}_l})\odot{\mathbf{z}}_{h,l}[n;P])$.
Then, the MMSE combining vector for the $k$th user's signal is $\mathbf{R}_{yy\vert \hat{\mathbf{H}}}^{-1}[n]\hat{\mathbf{g}}_{k}[n]$, and the decoded signal for $k$th user is given as \vspace{-4pt}
\begin{align}
&{r}_k[n]=\hat{\mathbf{g}}_{k}^H[n]\mathbf{R}^{-1}_{yy\vert \mathbf{\hat{H}}}[n]\hat{\mathbf{g}}_{k}[n]s_k[n]+\sum_{\substack{l=1\\l\neq k}}^K \hat{\mathbf{g}}_{k}^H[n]\mathbf{R}^{-1}_{yy\vert \mathbf{\hat{H}}}[n]\nonumber \\
&\times\hat{\mathbf{g}}_{l}[n]s_l[n]+\sum_{l=1}^K \hat{\mathbf{g}}_{k}^H[n]\mathbf{R}^{-1}_{yy\vert \mathbf{\hat{H}}}[n]\tilde{\mathbf{g}}_{l}[n]s_l[n]+\sum_{l=1}^K \hat{\mathbf{g}}_{k}^H[n] \nonumber \\
&\times \mathbf{R}^{-1}_{yy\vert \mathbf{\hat{H}}}[n]\boldsymbol{\xi}_{h,l}[n;P]s_l[n]\!+\!\sqrt{N_0}\hat{\mathbf{g}}_{k}^H[n]\mathbf{R}^{-1}_{yy\vert \mathbf{\hat{H}}}[n]\mathbf{w}[n]. 
\label{eq:r_signal}
\end{align}

Since the combining vector is based on the MMSE channel estimate, we can use the worst-case noise theorem~\cite{HandH}, and treat interference as noise to get our main result, {namely, the deterministic equivalent of the uplink achievable SINR.}

\begin{theorem}
The deterministic equivalent of the uplink achievable SINR for the $k$th user at the $n$th {instant}, conditioned on the spatial locations of the users, is given by 
\begin{equation}
\gamma_k[n]-\frac{\eta_{s,k}[n]}{\eta_{1,k}[n]+\eta_{2,k}[n]+\eta_{3,k}[n]+\eta_{w,k}}\xrightarrow{a.s.}0,
\end{equation}
where $\eta_{s,k}[n]$ is the desired signal power, $\eta_{1,k}[n]$ corresponds to the residual inter-user  interference after MMSE combining, $\eta_{2,k}[n]$ to the interference due to channel estimation errors, $\eta_{3,k}[n]$ the interference due to channel aging, and $\eta_{w,k}$ is due to AWGN. These can be expressed as 
\begin{equation} 
\eta_{s,k}[n]=N\mathcal{E}^2_{usk}\left(\sum_{m=1}^M\zeta_{mk}[n-P]\varphi_{mk}[n]\right)^2
 \label{eq:etaskn}
\end{equation}
with $\zeta_{mk}[n-P]=\beta_{mk}a^2_{mk}\rho^2_{mk}[n-P]$,
\begin{equation}
	\varphi_{mk}[n]\triangleq \left({\sum_{\substack{l=1\\l\neq k}}^{K}\mathcal{E}_{usl}\frac{{ \zeta_{ml}[n-P]} }{ (1+e_{k,l}[n])} +\psi_m} \right)^{-1} ,
	\label{eq:varphi_umk} 
\end{equation}
and $e_{k,l}[n]$ is iteratively computed as 
$
	e_{k,l}^{(t)}[n]=  N \sum_{m=1}^M\frac{ \zeta_{ml}[n-P]\mathcal{E}_{usl}} {  \sum_{i=1;i \ne k}^{K} \frac{ \zeta_{mi}[n-P]\mathcal{E}_{usi}}{1+e^{(t-1)}_{k,i}[n]} + \psi_m}
	\label{eq:eu}
$
with the initialization $e_{k,i}^{(0)}[n]=\frac{1}{N_0}$. 
\begin{equation} \label{eq:eta2kn}
\eta_{2,k}[n]=\mathcal{E}_{usk}\sum_{l=1}^K\mathcal{E}_{usl}\sum_{m=1}^M\varphi'_{mk}[n]\check{\zeta}_{ml}[n-P], 
\end{equation}
with $\check{\zeta}_{ml}[n-P]=\beta_{ml}\bar{a}^2_{ml}\rho^2_{ml}[n-P]$,  
\begin{equation}
	\varphi_{mk}'[n]=\varphi_{mk}^2[n]\left(\zeta_{mk}[n-P]+\sum_{\substack{l=1\\l\neq k}}^{K}\mathcal{E}_{usl}\frac{{ \zeta_{ml}[n-P]} e'_{k,l}[n] }{ (1+e_{k,l}[n])} \right ), 
\end{equation}
$e'_{k,l}[n]=[(\mathbf{I}_{K-1}-\mathbf{J})^{{-1}}\mathbf{u}]_l$,
$[\mathbf{J}]_{pq}=\frac{\mathcal{E}_{usp}\mathcal{E}_{usq}\text{Tr}\{\boldsymbol{\Xi}_p\mathbf{T}(\boldsymbol{\Psi})\boldsymbol{\Xi}_q\mathbf{T}(\boldsymbol{\Psi})\}}{1+e_{k,l}[n]}$,  ${u}_p=\mathcal{E}_{usp}\text{Tr}\{ \boldsymbol{\Xi}_p[n]\mathbf{T}(\boldsymbol{\Psi})\boldsymbol{\Xi}_k[n]\mathbf{T}(\boldsymbol{\Psi}) \}$, 
$ 
	\mathbf{T}(\boldsymbol{\Psi})=\left(\text{diag}(\boldsymbol{\varphi}_{k}[n])\right)\otimes\mathbf{I}_N
$
\begin{equation} \label{eq:eta3kn}
\eta_{3,k}[n]=\mathcal{E}_{usk}\sum_{l=1}^K\mathcal{E}_{usl}\sum_{m=1}^M\varphi'_{mk}[n]\dot{\zeta}_{ml}[n-P]
\end{equation}
with $\dot{\zeta}_{ml}[n-P]=\beta_{ml}\bar{\rho}_{ml}^2[n-P]$,
\begin{equation} \label{eq:etawkn}
\eta_{w,k}[n]=N_0\mathcal{E}_{usk}\sum_{m=1}^M\varphi'_{mk}[n],
\end{equation}
\begin{equation}\label{eq:eta1kn}
\eta_{1,k}[n]={\mathcal{E}_{usk}}\sum_{\substack{l=1\\l\neq k}}^K\mathcal{E}_{usl} \sum_{m=1}^M \zeta_{mk}[n-P]\zeta_{ml}[n-P] {\dot{\epsilon}_{kl}[n]}, 
\end{equation}
where
\vspace{-5pt}
\begin{align}
\dot{\epsilon}_{kl}[n]&=\dot{\varphi}'_{kl} [n] + 
\frac{\lvert \sum_{m=1}^M \mathcal{E}_{usl} \zeta_{ml}[n-P]  \rvert^2  N^2\dot{\varphi}^2_{kl}[n]\dot{\varphi}'_{kl}[n]  }{\lvert 1+ \sum_{m=1}^M \mathcal{E}_{usl} \zeta_{ml}[n-P]  N\dot{\varphi}_{kl}[n]\rvert^2 }  \nonumber \\&-2\Re \left \{  \frac{ \sum_{m=1}^M \mathcal{E}_{usl} \zeta_{ml}[n-P]    N\dot{\varphi}_{kl}[n]\dot{\varphi}'_{kl}[n] }{ 1+ \sum_{m=1}^M \mathcal{E}_{usl} \zeta_{ml}[n-P]  N \dot{\varphi}_{kl}[n]} \right \} 
\label{eq:eps_def}
\end{align}
with $\Re\{.\}$ denoting the real part of a complex number, and 
\begin{equation}
	\dot{\varphi}_{mkl}[n]\triangleq \left({\sum_{\substack{p=1\\p\neq k,l}}^{K}\mathcal{E}_{usp}\frac{{ \zeta_{mp}[n-P]} }{ (1+e_{k,p}[n])} +\psi_m} \right)^{-1} ,
\label{eq:varphi2} 
\end{equation}
$\dot{e}_{k,lp}[n]$ is iteratively computed as 
$	\dot{e}_{k,lp}^{(t)}[n]=  N\sum_{m=1}^M\frac{ \zeta_{ml}[n-P]\mathcal{E}_{usl}} {  \sum_{i=1;i \ne k}^{K} \frac{ \zeta_{mi}[n-P]\mathcal{E}_{usi}}{1+e^{(t-1)}_{k,i}[n]} + \psi_m}$
with the initialization $e_{mk,i}^{(0)}[n]=\frac{1}{N_0}$. Also, 
\begin{equation}
	\dot{\varphi}_{mkl}'[n]\!=\!\dot{\varphi}_{mkl}^2[n]\!\!\left(\!\!\zeta_{mk}[n-P]\!+\!\sum_{\substack{l=1\\l\neq k}}^{K}\frac{\mathcal{E}_{usl}{ \zeta_{ml}[n-P]} \dot{e}'_{k,l}[n] }{ (1+\dot{e}_{k,l}[n])} \right ), 
\end{equation}
with
$\dot{e}'_{k,l}[n]=[(\mathbf{I}_K-\mathbf{\dot{J}})\mathbf{\dot{u}}]_l$,
$[\mathbf{\dot{J}}]_{pq}=\frac{\mathcal{E}_{usp}\mathcal{E}_{usq}\text{Tr}\{\boldsymbol{\Xi}_p\mathbf{T}(\boldsymbol{\Psi})\boldsymbol{\Xi}_q\mathbf{T}(\boldsymbol{\Psi})\}}{1+e_{k,l}[n]}$,  $\dot{u}_p=\mathcal{E}_{usp}\text{Tr}\{ \boldsymbol{\Xi}_p[n]\mathbf{T}(\boldsymbol{\Psi})\boldsymbol{\Xi}_l[n]\mathbf{T}(\boldsymbol{\Psi} \}$.
\end{theorem}

\begin{proof}

We will soon post an updated version of this document with a detailed proof.
\end{proof}
\color{black}

{In the special case where the channel correlation coefficient $\rho[n]$ is independent of the UE and AP {indices}, it can be shown that $\eta_{s,k}[n],\eta_{1,k},$  and $\eta_{2,k}[n]$ are proportional to $\rho^4[n-P]$, while  $\eta_{3,k}[n]$ is proportional to ($\rho^2[n-P]-\rho^4[n-P]$). Since $\rho[n]$ is a decreasing function of $n$, their overall effect is an increase in interference and a decrease in the  SINR with time.} 
\section{Numerical Results}
In this section, we perform from Monte-Carlo simulations to corroborate the analytical results on the SINR of CF-mMIMO systems under channel aging. We consider a unit area including K=16 UEs served by N=256 single antenna APs. The signaling bandwidth is 5 MHz and the carrier frequency is 5 GHz. The frame duration is T=1024 channel uses. For all the experiments, we assume that the transmit SNR for both data and pilots is $20$~dB. For the multi-slope path loss model in \eqref{eq:multi_slope}, we assume $L=2$, $d_0=0.1$ units and $d_1=\sqrt{2}$ units while $\eta_0=0$ and $\eta_1=3$. To segregate the effects of pilot contamination and channel aging, we consider orthogonal pilots, transmitted over a duration of $16$ channel uses. 

For the computation of the average SINR at a given UE, the UEs and APs are dropped uniformly at random locations over a square with unit area. For each spatial realization of UEs and APs, $100$ independent channel realizations are used to compute the average uplink SINR achieved by a UE. This SINR is averaged over $100$ independent spatial realizations of the UE and AP locations. Thus, the average SINR is computed by averaging over 10,000 independent channel instantiations.

In Fig.~\ref{fig:CF_vs}, we plot the average per user SINR at the CPU as a function of the time index for different UE velocities with $\Delta=0$, i.e., the channel ages at the same rate at all APs. We see that the theoretical and simulated curves match perfectly for CF-mMIMO. {In the figure, theoretical and simulated curves are represented by the lines (no markers) and the  markers (no lines), respectively}. Also, we also compare the relative effects of channel aging on CF-mMIMO against those on cellular mMIMO and small cells. In the case of cellular mMIMO, we consider a single BS at the cell center equipped with $N=256$ antennas. For the case of small cells, we consider {$M=256$} single-antenna APs deployed over the area of interest (the same as in the CF case), with each UE associated with its nearest AP, under the transmit power assumptions considered in~\cite{Ngo_TWC_2019}. The theoretical expressions for the SINR achieved in these two systems can be derived in a similar manner as the expressions presented in this paper. We omit the details due to lack of space. CF-mMIMO achieves much higher SINR than both cellular mMIMO and small cells. Moreover, we observe that the impact of higher mobility on CF is more severe than that on cellular mMIMO or small cells. 
\begin{figure}[t!]
	\centering
	\includegraphics[width=0.45\textwidth]{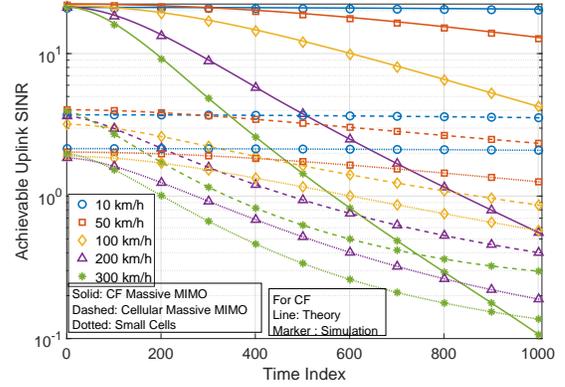}
	\caption{Uplink achievable SINR as a function of time for different UE velocities and system configurations.}\label{fig:CF_vs}
\end{figure}

In Fig.~\ref{fig:vsDelta}, we study effect of variation in the relative speed of the UE with respect to different APs for different numbers of APs. We do this by varying the parameter $\Delta$ described earlier under the assumption that the UEs are moving at an average speed of $v_{\max}=300$~km/h, and for different numbers of APs. We depict the average SINR of the CF-mMIMO system against the time index. Interestingly, a larger value of $\Delta$ mitigates the effects of channel aging. In addition, {we see that the relative loss in the SINR due to channel aging remains approximately unchanged as the number of APs is increased. In other words, contrary to intuition, an increase in the system dimension does not worsen the effects of aging. In other words, the benefits of higher dimensions offsets the degradation due to aging.} 
\begin{figure}[t!]
	\centering
	\includegraphics[width=0.4\textwidth]{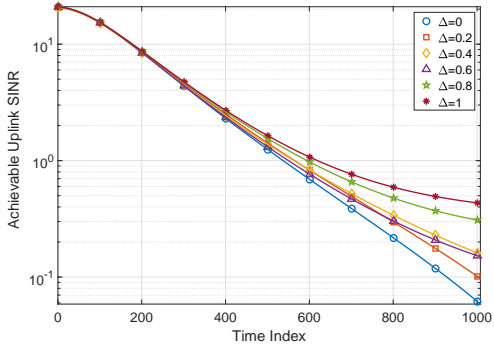}
	\caption{Uplink achievable SINR as a function of sample index at a UE velocity of $300$~km/h for different values of $\Delta$. }\label{fig:vsDelta}
\end{figure}

In Fig.~\ref{fig:Delta_vs_v}, we illustrate the average  SINR per user at the $1024$th sample as a function of the UE velocity for different values of $\Delta$. Also, we compare the achievable SINR against cellular mMIMO and small cells. We observe that at low to medium user mobility, CF-mMIMO significantly outperforms both cellular mMIMO and small cells. However, at high mobilities, the performance of CF-mMIMO becomes comparable to, if not worse than both cellular mMIMO and small cells. This strengthens the observation made in~\cite{Lett_aging_NOMA}, that the effect of channel aging on wireless communication systems is highly dependent on the system model, and warrants careful analysis. 
\begin{figure}[t!]
	\centering
	\includegraphics[width=0.4\textwidth]{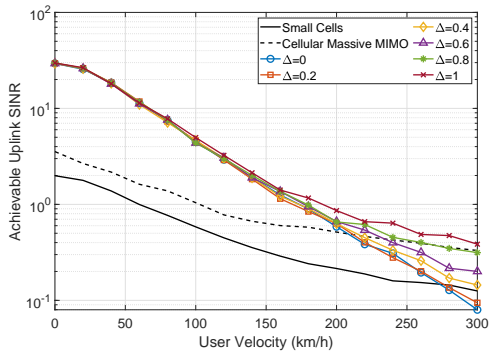}
	\caption{Uplink achievable SINR at the $1024$th sample, as a function of the UE velocity, for different values of $\Delta$. }\label{fig:Delta_vs_v}
\end{figure}
\vspace{-10pt}
\section{Conclusions and Future Work}
In this paper, we studied the effect of channel aging {on the  uplink of a general} CF-mMIMO system. We {elaborated on} the effects of different relative speeds of the UEs with respect to the APs, and derived the {deterministic equivalent SINR} under centralized MMSE combining at the CPU. We {observed} that CF-mMIMO {systems offer much better uplink SINR compared to cellular mMIMO systems and small cells at low user mobilities. However, at high velocities and long frame durations, the relative impact of aging is higher on CF-mMIMO systems.} We also note that unlike the Jakes’ model, there is no well-accepted model for characterizing the effects of mobility on mmWave channels. Therefore, the study of the effect of user mobility in mmWave massive MIMO requires well designed measurement campaigns to quantify the temporal correlation, variation in the angle of arrival/departures, and other channel characteristics. Hence, the study of channel aging in CF systems operating at the mmWave bands could be the topic for
future work.

 \vspace{-10pt}
\bibliographystyle{IEEEtran}
\bibliography{bibJournalList,fading,bib_DCP}
\end{document}